\documentclass{article}
\usepackage{amsmath,amsthm,amssymb}
\newtheorem{theorem}{Theorem}[section]

\newtheorem{lemma}{Lemma}[section]
\newtheorem{proposition}{Proposition}[section]
\theoremstyle{definition}

\theoremstyle{remark}
\newtheorem{remark}{Remark}[section]
\theoremstyle{definition}
\newtheorem{problem}{Problem}

\numberwithin{equation}{section}
\allowdisplaybreaks[3]
\newcommand\R{{\mathbb R}}
\newcommand\X{{\R^d}}
\newcommand\N{{\mathbb N}}
\newcommand\E{{\mathbb E}}
\newcommand\D[1][{}]{{\mathfrak D}^{^{#1}}}
\newcommand\DC{{\mathfrak D}^c}
\newcommand\K{{\mathcal K}}
\newcommand\M{{\mathcal M}}
\newcommand\B{{\mathcal B}}
\newcommand\Ga{\Gamma}
\newcommand\ga{\gamma}
\newcommand\La{\Lambda}
\newcommand\la{\lambda}

\newcommand\1{{1 \!\! 1}}
\newcommand{\step}[1]{\par${}\quad\mathbf{#1)}\quad$}
\newcommand\LC{L_{\mathtt{CM}}}
\newcommand\Lint{L_{\mathtt{int}}}
\newcommand\hLint{\hat{L}_{\mathtt{int}}}
\newcommand\dist{{\mathrm{dist}}}

\begin{document}

\title{On two-component contact model in continuum with~one~independent component}

\author{\textbf{Denis O.~Filonenko} \\
{\small National University ``Kyiv-Mohyla Academy''; 2 Skovoroda
str., 03070, Kyiv, Ukraine}\\ {\small denfil@ukr.net} \and \textbf{Dmitri L.~Finkelshtein} \\
{\small Institute of Mathematics, National Academy of Sciences of
Ukraine, Kiev, Ukraine}\\ {\small fdl@imath.kiev.ua} \and
\textbf{Yuri G.~Kondratiev} \\ {\small Fakult\"at f\"ur Mathematik,
Universit\"at Bielefeld, D 33615 Bielefeld, Germany}\\ {\small
Forschungszentrum BiBoS, Universit\"at Bielefeld, D 33615 Bielefeld, Germany}\\
{\small National University ``Kyiv-Mohyla Academy'', Kiev, Ukraine}\\
{\small kondrat@mathematik.uni-bielefeld.de}}

\maketitle



\begin{abstract}
Properties of a contact process in continuum for a system of two
type particles one type of which is independent are considered. We
study dynamics of the first and second order correlation functions,
their asymptotics and dependence on parameters of the~system.
\end{abstract}

\newpage

\section{Preliminaries}
The configuration space $\Ga :=\Ga _{\R^d}$ over $\R^d$,
$d\in\N$, is defined as the set of all locally finite subsets of
$\R^d$,
\begin{equation}
\Ga :=\left\{ \ga \subset \R^d:\left| \ga_\La\right|
<\infty \text{ for every compact }\La\subset \R^d\right\} ,
\end{equation}
where $\left| \cdot \right|$ denotes the cardinality of a set and
$\ga_\La := \ga \cap \La$. As usual we identify each
$\ga \in \Ga $ with the non-negative Radon measure $\sum_{x\in
\ga }\delta_x\in \mathcal{M}(\R^d)$, where $\delta_x$ is
the Dirac measure with unit mass at $x$,
$\sum_{x\in\varnothing}\delta_x$ is, by definition, the zero measure,
and $\mathcal{M}(\R^d)$ denotes the space of all
non-negative Radon measures on the Borel $\sigma$-algebra
$\mathcal{B}(\R^d)$. This identification allows to endow
$\Ga $ with the topology induced by the vague topology on
$\mathcal{M}(\R^d)$, i.e., the weakest topology on $\Ga$
with respect to which all mappings
\begin{equation*}
\Ga \ni \ga \longmapsto \langle f,\ga\rangle :=
\int_{\R^d}f(x)d\ga(x)=\sum_{x\in \ga }f(x),\quad
f\in C_0(\R^d),
\end{equation*}
are continuous. Here $C_0(\R^d)$ denotes the set of all continuous functions
on $\R^d$ with compact support. We denote by $\mathcal{B}(\Ga )$ the
corresponding Borel $\sigma$-algebra on $\Ga$.

Let us now consider the space of finite configurations
\begin{equation*}
\Ga_0 := \bigsqcup_{n=0}^\infty \Ga^{(n)},
\end{equation*}
where $\Ga^{(n)} := \Ga^{(n)}_{\R^d} := \{ \ga\in \Ga:
\vert \ga\vert = n\}$ for $n\in \N$ and $\Ga^{(0)} :=
\{\varnothing\}$. For $n\in \N$, there is a natural bijection between
the space $\Ga^{(n)}$ and the symmetrization
$\widetilde{(\R^d)^n}\diagup S_n$ of the set $\widetilde{(\R^d)^n}:=
\{(x_1,...,x_n)\in (\R^d)^n: x_i\not= x_j \hbox{ if } i\not= j\}$
under the permutation group $S_n$ over $\{1,...,n\}$ acting on
$\widetilde{(\R^d)^n}$ by permuting the coordinate indexes. This
bijection induces a metrizable topology on $\Ga^{(n)}$, and we
endow $\Ga_0$ with the topology of disjoint union of topological
spaces. By $\mathcal{B}(\Ga^{(n)})$ and $\mathcal{B}(\Ga_0)$
we denote the corresponding Borel $\sigma$-algebras on
$\Ga^{(n)}$ and $\Ga_0$, respectively.

Given a constant $z>0$, let $\la_z$ be the
Lebesgue-Poisson measure
$$
\la_z:=\sum_{n=0}^\infty \frac{z^n}{n!} m^{(n)},
$$
where each $m^{(n)}$, $n\in \N$, is the image measure on $\Gamma^{(n)}$ of
the product measure $dx_1...dx_n$ under the mapping
$\widetilde{(\R^d)^n}\ni (x_1,...,x_n)\mapsto\{x_1,...,x_n\}\in \Gamma^{(n)}$.
For $n=0$ we set $m^{(0)}(\{\varnothing\}):=1$.

We proceed to consider the $K$-transform \cite{Le73}, \cite{Le75a},
\cite{Le75b}, \cite{KoKu99}, that is, a mapping which maps functions
defined on $\Ga_0$ into functions defined on the space $\Ga$. Let
$\mathcal{B}_c(\R^d)$ denote the set of all bounded Borel sets in
$\R^d$, and for any $\La\in \mathcal{B}_c(\R^d)$ let $\Ga_\La :=
\{\eta\in \Ga: \eta\subset \La\}$. Evidently $\Ga_\La =
\bigsqcup_{n=0}^\infty \Ga_\La^{(n)}$, where $\Ga_\La^{(n)}:=
\Ga_\La \cap \Ga^{(n)}$ for each $n\in \N_0$, leading to a situation
similar to the one for $\Ga_0$, described above. We endow $\Ga_\La$
with the topology of the disjoint union of topological spaces and
with the corresponding Borel $\sigma$-algebra
$\mathcal{B}(\Ga_\La)$.

Given a $\mathcal{B}(\Ga_0)$-measurable function $G$ with local
support, that is, $G\!\!\upharpoonright
_{\Ga\setminus\Ga_\La}\equiv 0$ for some $\La \in
\mathcal{B}_c(\R^d)$, the $K$-transform of $G$ is a mapping
$KG:\Ga\to\R$ defined at each $\ga\in\Ga$ by
\begin{equation}
(KG)(\ga ):=\sum_{\eta \Subset \ga }
G(\eta ),
\label{Eq2.9}
\end{equation}
where $\eta\Subset\ga$ means that $\eta\subset\ga$ and $\vert\eta\vert < \infty$. Note that for every such function $G$ the sum in (\ref{Eq2.9}) has only a
finite number of summands different from zero, and thus $KG$ is a well-defined
function on $\Ga$. Moreover, if $G$ has support described as before, then
the restriction $(KG)\!\!\upharpoonright _{\Ga _\La }$ is a
$\mathcal{B}(\Ga_\La)$-measurable function and
$(KG)(\ga)=(KG)\!\!\upharpoonright _{\Ga _\La }\!\!(\ga_\La)$
for all $\ga\in\Ga$, i.e., $KG$ is a cylinder function.

Let now $G$ be a bounded $\mathcal{B}(\Ga_0)$-measurable function
with bounded support, that is, $G\!\!\upharpoonright _{\Ga
_0\backslash \left(\bigsqcup_{n=0}^N\Ga _\La ^{(n)}\right)
}\equiv 0$ for some $N\in\N_0, \La \in \mathcal{B}_c(\R^d)$. In
this situation, for each $C\geq \vert G\vert$ one finds $\vert
(KG)(\ga)\vert\leq C(1+\vert\ga_\La\vert)^N$ for all
$\ga\in\Ga$. As a result, besides the cylindricity property,
$KG$ is also polynomially bounded. In the sequel we denote the space
of all bounded $\mathcal{B}(\Ga_0)$-measurable functions with
bounded support by $B_{bs}(\Ga_0)$. It has been shown in
\cite{KoKu99} that the $K$-transform is a linear isomorphism which
inverse mapping is defined on cylinder functions by
\begin{equation}
\left( K^{-1}F\right) (\eta ):=\sum_{\xi \subset \eta }(-1)^{|\eta
\backslash \xi |}F(\xi ),\quad \eta \in \Ga _0.
\end{equation}

\section{The description of problem and main results}
\subsection{Basic facts and notations}
Two-component contact process in $\R^d$ describes
a birth-and-death stochastic dynamics of
a infinite system of two type particles.
Such system may be interpreted as pair of
configurations in $\R^d$ as well as one configuration of marked particles
that means that each particle has mark (spin) $+1$ or $-1$. The first interpretation sometimes
is more useful but we should additionally
assume that these two configurations don't
interact.

Let us give the rigorous definitions. Consider
two copies of the space $\Ga$: $\Ga^+$
and $\Ga^-$. Let
\begin{equation}
\Ga^2:=\Big\{(\ga^+,\ga^-)\in\Ga^+\times\Ga^-
 :  \ga^+\cap\ga^-=\varnothing\Big\}.
\end{equation}
Any configuration $\ga:=(\ga^+,\ga^-)\in\Ga^2$
may be identified with marked configuration
\begin{equation*}
\hat{\ga}= \big\{(x,\sigma_x) : x\in\ga^+\cup\ga^-,
\sigma_x=\1_{x\in\ga^+}-\1_{x\in\ga^-}\big\}\in\hat{\Ga},
\end{equation*}
since $\ga^+\sqcup\ga^-\in\Ga$. Here $\hat{\Ga}$
is the space of all marked configurations
in $\R^d$ with marks equal to $\pm 1$. One can induce topology on $\Ga^2$
from the weakest topology on $\hat{\Ga}$ such that all functions
\begin{equation*}
\hat{\Ga}\ni\hat{\ga}\longmapsto \sum_{(x,\sigma_x)\in\hat{\ga}}\hat{
f}((x,\sigma_x))\in\R
\end{equation*}
are continuous for all $\hat{f}\in C_0(\R^d\times\{-1;1\})$.
Clearly, in this induced topology on $\Ga^2$
all functions
\begin{equation*}
\Ga^2\ni\ga=(\ga^+,\ga^-)\longmapsto \sum_{x\in\ga^+}f(x)
+ \sum_{y\in\ga^-}g(y)\in\R
\end{equation*}
will be continuous for any $f,g\in C_0(\R^d)$.

On the other hand this topology may be induced from the topology on product
$\Ga^+\times\Ga^-$. Let $\B(\Ga^2):=\B(\Ga^+)\times\B(\Ga^-)$ be the~corresponding $\sigma$-algebra.

Let us now consider the space of finite configurations. Consider
two copies of the space $\Ga_0$: $\Ga^+_0$
and $\Ga^-_0$. Let
\begin{equation}
\Ga^2_0:=\Big\{(\eta^+,\eta^-)\in\Ga^+_0\times\Ga^-_0 : \eta^+\cap\eta^-=\varnothing\Big\}.
\end{equation}
Again one can consider the topology on $\Ga_0^2$ induced by the product-topology. By $\B(\Ga^2_0):=\B(\Ga^+_0)\times\B(\Ga^-_0)$ we denote
the~corresponding $\sigma$-algebra.

We will say that a function $G:\Ga_0^2\to\R$ is a bounded function
with bounded support if for any $(\eta^+,\eta^-)\in\Ga_0^2$
\begin{equation*}
G(\cdot,\eta^-)\in B_{bs}(\Ga_0^+), \quad
G(\eta^+,\cdot)\in B_{bs}(\Ga_0^-).
\end{equation*}
Class of all such functions we denote by $B_{bs}(\Ga_0^2)$.

For any $G\in B_{bs}(\Ga_0^2)$ one can define the $\K$-transform of $G$ as mapping $\K G:\Ga^2\to\R$ defined at each $\ga=(\ga^+,\ga^-)\in\Ga^2$ by
\begin{equation}
(\K G)(\ga)=\sum_{\substack{\eta^+\Subset\ga^+ \\ \eta^-\Subset\ga^-}}
G(\eta^+,\eta^-).
\end{equation}
On the other hand if $\1^\pm$ are unit operators on functions on $\Ga^\pm_0$
and $K^+:=K\otimes\1^-$, $K^-:=\1^+\otimes K$ then
\begin{equation*}
\K=K^+K^-=K^-K^+.
\end{equation*}
Hence, $\K G <\infty$ and $\K G$ is cylinder function on both variables.

Moreover,
$\K G$ is polynomially bounded: for the proper $C>0$, $\La\in\B_c(\R^d)$, $N\in\N$
\begin{equation*}
\vert (\K G)(\ga) \vert \leq C (1+\vert \ga^+_\La \vert)^N (1+\vert \ga^-_\La \vert)^N.
\end{equation*}
The inverse mapping is defined
on cylinder (on both variables) functions by
\begin{equation}
(\K^{-1} F)(\eta):= \sum_{\substack{\xi^+\subset\eta^+ \\ \xi^-\subset\eta^-}}
(-1)^{\vert \eta^+\setminus\xi^+\vert + \vert \eta^-\setminus\xi^-\vert}F(\xi^+,\xi^-), \quad \eta=(\eta^+,\eta^-)\in
\Ga^2_0.
\end{equation}

Let $\mu$ be a probability measure on
$\Big(\Ga^2,\B\bigl(\Ga^2\bigr)\Big)$ (we denote class of the all such measures by $\M^1\bigl(\Ga^2\bigr)$).
The function $k_\mu:\Ga^2_0\to\R$ is called a~{\it correlation
function} of~the~measure $\mu$ if for~any $G\in B_{bs}(\Ga^2_0)$
\begin{equation}
\int_{\Ga^2} (\K G)(\ga) d\mu(\ga) =
\int_{\Ga^2_0} G(\eta^+,\eta^-) k_\mu(\eta^+,\eta^-)
d\la_1(\eta^+) d\la_1(\eta^-).
\end{equation}

\subsection{Description of model}
Let us consider the generator $L$ of two-component contact process with one independent component. This generator is well-defined at least on cylindric
functions on $\Ga^2$ and has the~following form:
\begin{equation}
L=\LC ^{+}+\LC ^{-}+\Lint ^{+}.
\end{equation}
Here $\LC ^{+}$ is the generator of the one-component contact model of
($+$)-system, $\LC ^{-}$ is the analogous generator of ($-$)-system, $\Lint ^{+}$ is interaction term that describes birth of ($+$%
)-particles under influence of ($-$)-particles. Namely,
\begin{align*}
\left( \LC ^{+} F\right) ( \ga ^{+},\ga ^{-}) &=\sum_{x\in \ga
^{+}}\left[ F\left( \ga ^{+}\setminus x,\ga ^{-}\right)
-F( \ga ^{+},\ga ^{-}) \right] \\
&\quad+\la ^{+}\int_{\R^{d}}\left( \sum_{x' \in \ga ^{+}}a^{+}\left(
x-x' \right) \right) \left[ F\left(
\ga ^{+}\cup x,\ga ^{-}\right) -F( \ga ^{+},\ga ^{-}) \right] dx,\\
\left( \LC ^{-} F\right) ( \ga ^{+},\ga ^{-}) &=\sum_{y\in \ga
^{-}}\left[ F\left( \ga ^{+},\ga ^{-}\setminus y\right)
-F( \ga ^{+},\ga ^{-}) \right] \\
&\quad +\la ^{-}\int_{\R^{d}}\left( \sum_{y' \in \ga
^{-}}a^{-}\left( y-y' \right) \right) \left[ F\left(
\ga ^{+},\ga ^{-}\cup y\right) -F( \ga ^{+},\ga ^{-}) \right] dy,\\
\left( \Lint ^{+} F\right) ( \ga ^{+},\ga ^{-}) &=\la
\int_{\R^{d}}\left( \sum_{y\in \ga ^{-}}a\left( x-y\right) \right)
\left[ F\left( \ga ^{+}\cup x,\ga ^{-}\right) -F( \ga ^{+},\ga ^{-})
\right] dx.
\end{align*}
Constants $\la^+,\la^-,\la$ are positive, functions $a^+,a^-,a$ are
non-negative, even, integrable and normalised:
\begin{equation*}
\langle a^{+}\rangle = \langle a^{-}\rangle = \langle a \rangle = 1.
\end{equation*}
Here and in the sequel we use the following notation
\begin{equation*}
\langle f \rangle :=\int_{\R^{d}}f(x)dx, \quad f\in L^1(\R^d).
\end{equation*}

We also denote the Fourier
transform of such $f$ as $\hat{f}$:
\begin{equation*}
\hat{f}(p)=\int_\X e^{-i(p,x)}f(x)dx,
\end{equation*}
where $(\cdot,\cdot)$ is a scalar product in $\X$.

Next theorem is the partial case of the results obtained in \cite{FKS}.
\begin{theorem}
Let $d\geq 2$ and there exists constants
$A>0, \delta>2d$ such that
\begin{equation}\label{decayofa}
a^+(x) + a^-(x) + a(x) \leq \frac{A}{(1+\vert x \vert)^\delta}.
\end{equation}
Then there exists a Markov process $X_t$ on $\Ga^2$ with generator $L$.
\end{theorem}

We will always suppose also that
\begin{equation}\label{intoffour}
    \hat{a}, \hat{a}^+, \hat{a}^- \in L^1(\R^d).
\end{equation}

Hence, one has stochastic dynamics of configurations that implies
dynamics of measures, namely $\M^1\bigl(\Ga^2\bigr)\ni\mu_0 \mapsto
\mu_t\in\M^1\bigl(\Ga^2\bigr)$ such that
for any measurable bounded $F:\Ga^2\rightarrow\R$
\begin{equation*}
\int_{\Ga^2}F(\ga)d\mu_t(\ga):= \E \Biggl[\int_{\Ga^2}F(X_t^\ga)d\mu_0(\ga)\Biggr],
\end{equation*}
where process $X_t$ starts from $\ga\in\Ga^2$ (more precisely, $\ga$
belongs to proper support set, see \cite{FKS}).

This dynamics of measures implies dynamics of corresponding
correlation functions (if they exist). For obtain explicit
differential equations for this dynamics we should calculate so-called
descent operator $\hat{L}$ which defined on functions
$G\in B_{bs}(\Ga^2_0)$ by
\begin{equation}
\bigl( \hat{L}G \bigr) ( \eta ) = \bigl( (\K^{-1} L \K) G \bigr) ( \eta ),
\quad \eta\in\Ga^2_0.
\end{equation}
Next we should obtain the adjoint operator $\hat{L}^{\ast}$ (with respect
to measure $d\la_1d\la_1$):
\begin{multline}\label{adjoint}
\int_{\Ga _{0}^{2}}\hat{L}G ( \eta ^{+},\eta ^{-})
k ( \eta ^{+},\eta ^{-}) d\la_1( \eta
^{+}) d\la _1( \eta ^{-}) \\=\int_{\Ga
_{0}^{2}}G ( \eta ^{+},\eta ^{-}) \hat{L}^{\ast } k( \eta
^{+},\eta ^{-}) d\la_1 ( \eta ^{+} ) d\la_1
( \eta ^{-} ) .
\end{multline}
Then equations for time evolution of correlation function
will be following:
\begin{equation}\label{geneq}
\frac{\partial k_{t} ( \eta ^{+},\eta ^{-} ) }{\partial t}=\bigl(
\hat{L}^{\ast }k_{t}\bigr) ( \eta ^{+},\eta ^{-} ) .
\end{equation}

In the present article we concentrate our attention
on the correlation functions of the first and second orders:
\begin{equation}\label{cf12}
\begin{aligned}
k_t^+(x)&:=k_t(\{x\},\varnothing), &\quad x&\in\R^d;\\
k_t^-(y)&:=k_t(\varnothing,\{y\}), &\quad y&\in\R^d;\\
k_t^{++}(x_1,x_2)&:=k_t(\{x_1, x_2\},\varnothing), &\quad x_1, x_2&\in\R^d;\\
k_t^{+-}(x,y)&:=k_t(\{x\},\{y\}), &\quad x, y&\in\R^d;\\
k_t^{--}(y_1,y_2)&:=k_t(\varnothing, \{y_1, y_2\}), &\quad y_1, y_2&\in\R^d.
\end{aligned}
\end{equation}

The main subject for our studying will be explicit
expression for correlation functions of the first and second orders
and their asymptotic at $t\to\infty$.

\subsection{Problems and results}
In this subsection we state main problems and formulate results.
All proofs are presented in the next section.

First two results give explicit forms of~the~equation~\eqref{geneq}
for the~first and second order correlation functions~\eqref{cf12}.
\begin{proposition}\label{propeq1}
For any $x,y\in\R^d$
\begin{align*}
\frac{\partial k_{t}^{-}(y)}{\partial t} &=-k_{t}^{-}(y)+\la
^{-}\int_{\R^{d}}a^{-}(y-y' )k_{t}^{-}(y' )dy' ,
\\
\frac{\partial k_{t}^{+}(x)}{\partial t} &=-k_{t}^{+}(x)+\la
^{+}\int_{\R^{d}}a^{+}(x-x' )k_{t}^{+}(x' )dx' +\la
\int_{\R^{d}}a(x-y)k_{t}^{-}(y)dy
\end{align*}%
\end{proposition}

\begin{proposition}\label{propeq2}
For any $x,y,x_1,x_2,y_1,y_2\in\R^d$
\begin{align*}
\frac{\partial k_{t}^{--}(y_{1},y_{2})}{\partial t} &= \la
^{-}\int_{\R^{d}}a^{-}(y_{2}-y' )k_{t}^{--}(y_{1},y' )dy' +\la
^{-}\int_{\R^{d}}a^{-}(y_{1}-y' )k_{t}^{--}(y_{2},y' )dy'  \\
&\quad -2k_{t}^{--}(y_{1},y_{2}) +\la ^{-}a^{-}(y_{1}-y_{2})[k_{t}^{-}(y_{1})+k_{t}^{-}(y_{2})], \\
\frac{\partial k_{t}^{+-}(x,y)}{\partial t} &=\la
^{+}\int_{\R^{d}}a^{+}(x-x' )k_{t}^{+-}(x' ,y)dx'  +\la
^{-}\int_{\R^{d}}a^{-}(y-y' )k_{t}^{+-}(x,y' )dy'  \\
&\quad -2k_{t}^{+-}(x,y)   +\la a(x-y)k_{t}^{-}(y) + \la
\int_{\R^{d}}a(x-y' )k_{t}^{--}(y,y' )dy' , \\
\frac{\partial k_{t}^{++}(x_{1},x_{2})}{\partial t} &=\la
^{+}\int_{\R^{d}}a^{+}(x_{1}-x' )k_{t}^{++}(x_{2},x' )dx' +\la
^{+}\int_{\R^{d}}a^{+}(x_{2}-x' )k_{t}^{++}(x_{1},x' )dx'  \\
&\quad -2k_{t}^{++}(x_{1},x_{2}) +\la ^{+}a^{+}(x_{1}-x_{2})[k_{t}^{+}(x_{1})+k_{t}^{+}(x_{2})] \\
&\quad+\la \int_{\R^{d}}a(x_{1}-y)k_{t}^{+-}(x_{2},y)dy +\la
\int_{\R^{d}}a(x_{2}-y)k_{t}^{+-}(x_{1},y)dy .
\end{align*}%
\end{proposition}
Obviously, equations for ($-$)-system are independent. Recall that
such equations were studied in \cite{KoKuPi07}.

Let us formulate the main problem
for the first order correlation functions.
\begin{problem}
We should to study the asymptotic properties of the solutions
of equations from Proposition~\ref{propeq1} under following
initial conditions:
\begin{equation}\label{initval1}
k_{0}^{+}(x) =c^{+}+\psi ^{+}(x)\geq 0, \qquad
k_{0}^{-}(y) =c^{-}+\psi ^{-}(y)\geq \alpha^->0,
\end{equation}
where constants $c^{+}, c^{-}$ are positive, functions $\psi ^{+}, \psi ^{-}$ and their Fourier
transforms $\hat{\psi}^{+}, \hat{\psi}^{-}$ are
integrable on $\R^{d}$.
\end{problem}

Explicit expressions for solutions are in the next section.
The answer of the Problem~1  may be found in the next theorem.
\begin{theorem}\label{thm_as_1}
Let $d\geq  3$ and \eqref{decayofa}, \eqref{intoffour} hold. The
first correlation functions have the following asymptotic at $t\to
\infty$: \step{1} for any $y\in\R^d$
\begin{equation*}
k_{t}^{-}(y)\rightarrow \left\{
\begin{array}{@{\,}r@{\quad}l@{}}
0, & \text{if } \ \la ^{-}<1 \\
\infty , & \text{if } \ \la ^{-}>1
\end{array}
\right. ,
\end{equation*}
and in the case $\la ^{-}=1$
\begin{equation*}
k_{t}^{-}(y)\rightarrow c^{-};
\end{equation*}
\step{2} for any $x\in\R^d$
\begin{equation*}
k_{t}^{+}(x)\rightarrow \left\{
\begin{array}{@{\,}r@{\quad}l@{}}
0, & \text{if } \ \max \{\la ^{+},\la ^{-}\}<1 \\
\infty , & \text{if } \ \min \{\la ^{+},\la ^{-}\}\geq 1%
\end{array}%
\right. ,
\end{equation*}
next, in the case $1=\la ^{+}>\la ^{-}$
\begin{equation*}
k_{t}^{+}(x)\rightarrow c^{+}+\frac{\la c^{-}}{1-\la ^{-}},
\end{equation*}
and in the case $\la ^{+}<\la ^{-}=1$
\begin{equation*}
k_{t}^{+}(x)\rightarrow \frac{\la c^{-}}{1-\la ^{+}}.
\end{equation*}
\end{theorem}

Let us discuss this result. Of course, first part about the
independent $(-)$-system is the same as in \cite{KoKuPi07, KoSk06}.
It state that $\la^-=1$ is critical value; below of this value
$(-)$-system will degenerate at infinity, above of this value
$(-)$-system will grow (exponentially, see next section for
details). At this critical value $(-)$-system continues to be
stable.

$(+)$-system consists of two parts: independent
contact and influence from the side of~
$(-)$-system. If~$\max \{\la ^{+},\la ^{-}\}<1$
it means that independent part of $(+)$-system is sub-critical
(and~should disappear at infinity)
and additionally it has influence of disappearing
 $(-)$-system; naturally, such $(+)$-system  will
disappear. If $\min \{\la ^{+},\la ^{-}\}\geq
1$ it means that growing or stable independent part of $(+)$-system has
influence by stable or growing \mbox{$(-)$-system}, hence,
$(+)$-system will grow.

Let us concentrate our attention on two other
cases. If $\la ^{+}=1, \la ^{-}<1$ it means
that independent part of $(+)$-system is
stable and has influence by degenerating
$(-)$-system. As a result, $(+)$-system will
keep stability property but the limiting
value will have the initial value
of $(-)$-system which will disappearing at
infinity. Hence, $(+)$-system will have memory
about vanished $(-)$-system.

If $\la ^{+}<1, \la ^{-}=1$ it means that
degenerating independent part of $(+)$-system
has influence by stable $(-)$-system. In
result, $(+)$-system will stop disappearing
and become stable. But ``fare'' for this will
be absence of the initial value of
$(+)$-system in limit. Therefore, $(+)$-system
``will lost memory'' about its origin and ``remember''
only about origin of ``donor''.

In studying asymptotic of the second correlation
functions we concentrate our attention only
on this two cases when $(+)$-system will
be stable. For simplicity of computations
we consider translation invariant case only:
\begin{equation}\label{tr_inv}
\psi ^{+}=\psi ^{-}\equiv 0.
\end{equation}

\begin{problem}
We should to study the asymptotic properties of the solutions of equations from Proposition~\ref{propeq2}
under following initial conditions:
\begin{equation}\label{initval2}
\begin{aligned}
k_{0}^{++}(x_{1},x_{2}) &=c^{++}+\varphi ^{++}(x_{1}-x_{2})\geq
0, \\
k_{0}^{+-}(x,y) &=c^{+-}+\varphi ^{+-}(x-y)\geq  0, \\
k_{0}^{--}(y_{1},y_{2}) &=c^{--}+\varphi ^{--}(y_{1}-y_{2})\geq  0,
\end{aligned}
\end{equation}
where $c^{--},~c^{+-},~c^{++}$ are positive
constants and
 and functions $\varphi ^{--},~\varphi ^{+-},~\varphi ^{++}$ are
even functions which are integrable on $\R^{d}$ together with
their Fourier transforms
$\hat{\varphi}^{--},~\hat{\varphi}^{+-},~\hat{\varphi}^{++}$.
\end{problem}

Explicit expressions for solutions are also in the next section.
The answer of the Problem~2  may be found in the next theorem.

\begin{theorem} \label{thm_as_2} Let $d\geq  3$  and \eqref{decayofa}, \eqref{intoffour},
\eqref{tr_inv} hold. The second correlation functions have the
following asymptotic at $t\to \infty$: \step{1} let $\la
^{+}=1,~0<\la ^{-}<1$, then for any $x,y,x_1,x_2,y_1,y_2\in\R^d$
\begin{equation*}
\left\{
\begin{aligned}
k_{t}^{--}(y_{1},y_{2})&\rightarrow0, \\
k_{t}^{+-}(x,y)&\rightarrow0, \\
k_{t}^{++}(x_{1},x_{2})&\rightarrow \left( c^{++}-\frac{%
2\la c^{+-}}{\la ^{-}-1}+\frac{\la ^{2}c^{--}}{(\la
^{-}-1)^{2}}\right) +\Omega ^{++}(x_{1}-x_{2})<\infty;
\end{aligned}%
\right.
\end{equation*}

\step{2} let $\la ^{-}=1,~0<\la ^{+}<1$,
then for any $x,y,x_1,x_2,y_1,y_2\in\R^d$
\begin{equation*}
\left\{
\begin{aligned}
k_{t}^{--}(y_{1},y_{2})&\rightarrow c^{--}+\Xi
^{--}(y_{1}-y_{2})<\infty,  \\
k_{t}^{+-}(x,y)&\rightarrow\frac{\la c^{--}}{1-\la
^{+}}+\Xi ^{+-}(x-y)<\infty,  \\
k_{t}^{++}(x_{1},x_{2})&\rightarrow\frac{\la ^{2}c^{--}}{%
(1-\la ^{+})^{2}}+\Xi ^{++}(x_{1}-x_{2})<\infty;
\end{aligned}
\right.
\end{equation*}%
here functions $\Xi ^{--},\Xi ^{+-},\Xi ^{++}$ depend on initial
value $c^{-}$ only and function $\Omega ^{++}$ depends on initial
value $c^{+}$ only
 (of course,
they also depend on $\la,\la^\pm, a, a^\pm$).
\end{theorem}

The explicit expressions for limits will be presented in the next
section.

As we see, the situation with ``memory'' which we had for the first
correlation functions is the same for the second one: in the first
case $(+)$-system will obtain additional memory about vanished
$(-)$-system; in the second case $(+)$-system will have memory about
$(-)$-system only.

\begin{remark}
Note that if $c^{++}=(c^+)^2$, $c^{+-}=c^+c^-$,
$c^{--}=(c^-)^2$ then the previous theorems
show, in fact, that there exist
finite limits of so-called second order Ursell
functions $k_t^{++}-(k_t^{+})^2$, $k_t^{+-}-k_t^{+}k_t^{-}$,
$k_t^{--}-(k_t^{-})^2$.
\end{remark}

\section{Proofs}
In this section we present proofs of all
our results.

\subsection{Equations for time evolution of the correlation functions}
First of all we show how to obtain the equations from the Propositions~\ref{propeq1} and \ref{propeq2}.
We start from the explicit form of the descent
operator $\hat{L}$.

\begin{proposition}
Let $G\in B_{bs}(\Ga^2_0)$. Then for any
$\eta=(\eta^+,\eta^-)\in\Ga^2_0$
\begin{align*}
\left( \hat{L}G\right) ( \eta ^{+},\eta ^{-})
&=-\left( \left\vert \eta ^{+}\right\vert +\left\vert \eta
^{-}\right\vert
\right) G( \eta ^{+},\eta ^{-}) \\
&\quad+\la ^{+}\int_{\X}G\left( \eta ^{+}\cup x,\eta ^{-}\right)
\left(
\sum_{x' \in \eta ^{+}}a^{+}\left( x-x' \right) \right) dx \\
&\quad+\la ^{+}\int_{\X}\sum_{x' \in \eta ^{+}}G\left( \eta
^{+}\setminus x' \cup x,\eta ^{-}\right) a^{+}\left( x-x' \right) dx
\\
&\quad+\la ^{-}\int_{\X}G\left( \eta ^{+},\eta ^{-}\cup y\right)
\left(
\sum_{y' \in \eta ^{-}}a^{-}\left( y-y' \right) \right) dy \\
&\quad+\la ^{-}\int_{\X}\sum_{y' \in \eta ^{-}}G( \eta ^{+},\eta
^{-}\setminus y' \cup y) a^{-}\left( y-y' \right) dy \\
&\quad+\la \int_{\X}G\left( \eta ^{+}\cup x,\eta ^{-}\right) \left(
\sum_{y' \in \eta ^{-}}a\left( x-y' \right) \right) dx \\
&\quad+\la \int_{\X}\sum_{y' \in \eta ^{-}}G\left( \eta ^{+}\cup
x,\eta ^{-}\setminus y' \right) a\left( x-y' \right) dx
\end{align*}
\end{proposition}
\begin{proof}
Let us denote death and birth parts of
the operator $\LC^+$ by
\begin{align*}
(L_{d}^{+}F)(\ga^+,\ga^-) &:=\sum_{x\in \ga ^{+}}\left[ F\left( \ga
^{+}\setminus
x,\ga ^{-}\right) -F( \ga ^{+},\ga ^{-}) \right], \\
(L_{b}^{+}F)(\ga^+,\ga^-) &:=\la ^{+}\int_{\R^{d}}\left( \sum_{x'
\in \ga ^{+}}a^{+}\left( x-x' \right) \right) \left[ F\left( \ga
^{+}\cup x,\ga ^{-}\right) -F( \ga ^{+},\ga ^{-}) \right] dx.
\end{align*}
In the same way we denote death and birth parts of the operator $\LC^-$: $\LC^-=L_{d}^{-}+L_{b}^{-}$.
As a result,
\begin{equation*}
L =L_{d}^{+}+L_{b}^{+}+L_{d}^{-}+L_{b}^{-}+\Lint ^{+}.
\end{equation*}

Now we calculate image under $\K$-transform
of all this operators.
One has for any $\eta=(\eta^+,\eta^-)\in\Ga^2_0$
\begin{align*}
\left( \hat{L}_{b}^{+}G\right) ( \eta) &=\left( \K ^{-1} L_{b}^{+}\K^{+} G\right)
( \eta ) \\
&=\sum_{\xi ^{+}\subset \eta ^{+}}(-1)^{|\eta ^{+}\setminus \xi
^{+}|}\sum_{\xi ^{-}\subset \eta ^{-}}(-1)^{|\eta ^{-}\setminus \xi
^{-}|}\la ^{+}\int_{\R^{d}}\sum_{x' \in \xi
^{+}}a^{+}(x-x' ) \\
&\qquad \times \left( \sum_{\zeta ^{+}\subset \xi ^{+}\cup x}\sum_{\zeta
^{-}\subset \xi ^{-}}G(\zeta ^{+},\zeta ^{-})-\sum_{\zeta
^{+}\subset \xi
^{+}}\sum_{\zeta ^{-}\subset \xi ^{-}}G(\zeta ^{+},\zeta ^{-})\right) dx \\
&=\la ^{+}\int_{\X}\sum_{x' \in \eta ^{+}}G\left( \eta ^{+}\cup
x,\eta ^{-}\right) a^{+}\left( x-x' \right) dx \\
&\quad+\la ^{+}\int_{\X}\sum_{x' \in \eta ^{+}}G\left( \eta
^{+}\setminus x' \cup x,\eta ^{-}\right) a^{+}\left( x-x' \right)
dx,
\end{align*}%
analogously, we have that%
\begin{align*}
\left( \hat{L}_{b}^{-}G\right) ( \eta ^{+},\eta ^{-}) &=\la
^{-}\int_{\X}\sum_{y' \in \eta ^{-}}G\left(
\eta ^{+},\eta ^{-}\cup y\right) a^{-}\left( y-y' \right) dy \\
&\quad\quad+\la ^{-}\int_{\X}\sum_{y' \in \eta ^{-}}G\left( \eta
^{+},\eta ^{-}\setminus y' \cup y\right) a^{-}\left( y-y' \right)
dy.
\end{align*}

Next,
\begin{align*}
\left( \hLint^{+}G\right) ( \eta )
&=\left( \K ^{-1} \Lint^{+}\K^{+} G\right) ( \eta ) \\
&=\sum_{\xi ^{+}\subset \eta ^{+}}(-1)^{|\eta ^{+}\setminus \xi
^{+}|}\sum_{\xi ^{-}\subset \eta ^{-}}(-1)^{|\eta ^{-}\setminus \xi
^{-}|}\la \int_{\R^{d}}\sum_{y\in \xi ^{-}}a(x-y) \\
&\qquad\times \left( \sum_{\zeta ^{+}\subset \xi ^{+}\cup x}\sum_{\zeta
^{-}\subset \xi ^{-}}G(\zeta ^{+},\zeta ^{-})-\sum_{\zeta
^{+}\subset \xi
^{+}}\sum_{\zeta ^{-}\subset \xi ^{-}}G(\zeta ^{+},\zeta ^{-})\right) dx \\
&=\la \int_{\X}\sum_{y' \in \eta ^{-}}G\left( \eta ^{+}\cup
x,\eta ^{-}\right) a\left( x-y' \right) dx \\
&\quad+\la \int_{\X}\sum_{y' \in \eta ^{-}}G\left( \eta ^{+}\cup
x,\eta ^{-}\setminus y' \right) a\left( x-y' \right) dx.
\end{align*}%
Finally,
\begin{align*}
\left( \hat{L}_{d}^{-}G\right) ( \eta ) &=\left( \K ^{-1} L_{d}^{-}\K^{+} G\right)
( \eta ) \\
&=\sum_{\xi ^{+}\subset \eta ^{+}}(-1)^{|\eta ^{+}\setminus \xi
^{+}|}\sum_{\xi ^{-}\subset \eta ^{-}}(-1)^{|\eta ^{-}\setminus \xi ^{-}|} \\
&\qquad\times\sum_{y\in \xi ^{-}}\left( \sum_{\zeta ^{+}\subset \xi
^{+}}\sum_{\zeta ^{-}\subset \xi ^{-}\setminus y}G(\zeta ^{+},\zeta
^{-})-\sum_{\zeta ^{+}\subset \xi ^{+}}\sum_{\zeta ^{-}\subset \xi
^{-}}G(\zeta ^{+},\zeta
^{-})\right) \\
&=-\left\vert \eta ^{-}\right\vert G( \eta ^{+},\eta ^{-}),
\end{align*}%
and, analogously,
\begin{equation*}
\left( \hat{L}_{d}^{+}G\right) ( \eta ^{+},\eta ^{-}) =-\left\vert
\eta ^{+}\right\vert G( \eta ^{+},\eta ^{-}).
\end{equation*}%
The statement is proved.
\end{proof}

Now we should calculate the adjoint operator
$\hat{L}^{\ast }$.
\begin{proposition}\label{prop_adj_oper}
The adjoint operator $\hat{L}^{\ast }$ has the following form:
\begin{align*}
\left( \hat{L}^{\ast }k\right) ( \eta ^{+},\eta ^{-}) &=-\left(\left\vert \eta ^{+}\right\vert +\left\vert \eta ^{-}\right\vert
\right)k( \eta ^{+},\eta ^{-}) \\
&\quad+\la ^{+}\sum_{x\in \eta ^{+}}\sum_{x' \in \eta ^{+}\setminus
x}a^{+}(x-x' )k\left( \eta ^{+}\setminus x,\eta ^{-}\right) \\
&\quad+\la ^{+}\sum_{x\in \eta ^{+}}\int_{\R^{d}}a^{+}(x-x' )k\left(
\eta ^{+}\setminus x\cup x' ,\eta ^{-}\right) dx'
\\
&\quad+\la ^{-}\sum_{y\in \eta ^{-}}\sum_{y' \in \eta ^{-}\setminus
y}a^{-}(y-y' )k\left( \eta ^{+},\eta ^{-}\setminus y\right) \\
&\quad+\la ^{-}\sum_{y\in \eta ^{-}}\int_{\R^{d}}a^{-}(y-y' )k\left(
\eta ^{+},\eta ^{-}\setminus y\cup y' \right) dy'
\\
&\quad+\la \sum_{x\in \eta ^{+}}\sum_{y\in \eta ^{-}}a(x-y)k (\eta
^{+}\setminus x,\eta ^{-}) \\
&\quad+\la \sum_{x\in \eta
^{+}}\int_{\R^{d}}a(x-y)k\left( \eta ^{+}\setminus
x,\eta ^{-}\cup y\right) dy
\end{align*}
\end{proposition}

\begin{proof}
We may use
the following corollaries of the classical  Mecke formula (see, e.g., \cite{AKR1}):
\begin{multline*}
\int_{\Ga _{0}^{2}}\sum_{x\in \eta ^{+}}h_{+}(x,\eta
^{+},\eta ^{-})d\la_{1}( \eta ^{+}) d\la
_1( \eta ^{-}) \\
\shoveright{
=\int_{\Ga
_{0}^{2}}\int_{\R^{d}}h_{+}(x,\eta ^{+}\cup x,\eta
^{-})dxd\la_{1}( \eta ^{+}) d\la
_1( \eta ^{-}),}
\\
\shoveleft{
\int_{\Ga _{0}^{2}}\sum_{y\in \eta ^{-}}h_{-}(y,\eta
^{+},\eta ^{-})d\la_{1}( \eta ^{+}) d\la
_1( \eta ^{-})} \\
\shoveright{
=\int_{\Ga
_{0}^{2}}\int_{\R^{d}}h_{-}(y,\eta ^{+},\eta ^{-}\cup
y)dyd\la_{1}( \eta ^{+}) d\la
_1( \eta ^{-}),}
\\
\shoveleft{
\int_{\Ga _{0}^{2}}\sum_{x\in \eta ^{+}}\sum_{y\in \eta
^{-}}h(x,\eta ^{+},\eta ^{-})d\la_{1}( \eta ^{+}) d\la
_1( \eta ^{-})}\\
=\int_{\Ga
_{0}^{2}}\int_{\R^{d}}\int_{\R^{d}}h(x,\eta
^{+}\cup x,\eta ^{-}\cup y)dxdyd\la_{1}( \eta ^{+}) d\la
_1( \eta ^{-}).
\end{multline*}

Then one can obtain the explicit formula for the operator $\hat{L}^{\ast }$ directly
from definition~\eqref{adjoint}.
\end{proof}

As a result, the statements of the Propositions~\ref{propeq1} and
\ref{propeq2} are directly follow from the
Proposition~\ref{prop_adj_oper} and \eqref{geneq}--\eqref{cf12}.

\subsection{Solution of the equations for time evolution of the correlation
functions}
To solve the equations from the Propositions~\ref{propeq1}
and \ref{propeq2} using classical
perturbation method we rewrite these equations in the following
forms:
\begin{align}
\frac{\partial k_{t}^{-}(y)}{\partial t} &=(\la
^{-}-1)k_{t}^{-}(y)+\la ^{-}(L^{-}k_{t}^{-})(y), \label{eq1-}\\
\frac{\partial k_{t}^{+}(x)}{\partial t} &=(\la
^{+}-1)k_{t}^{+}(x)+\la ^{+}(L^{+}k_{t}^{+})(x)+\la
\int_{\R^{d}}a(x-y)k_{t}^{-}(y)dy,\label{eq1+}
\end{align}
where Markov-type generators $L^\pm$ are
defined on functions on $\X$ by
\begin{align*}
(L^{-}f)(y) &=\int_{\R^{d}}a^{-}(y-y' )[f(y' )-f(y)]dy',  \\
(L^{+}f)(x) &=\int_{\R^{d}}a^{+}(x-x' )[f(x' )-f(x)]dx';
\end{align*}
and for the second order correlation functions:
\begin{align}
\frac{\partial k_{t}^{--}(y_{1},y_{2})}{\partial t}
&=2k_{t}^{--}(y_{1},y_{2})(\la ^{-}-1)+\la
^{-}(L_{1}^{--}k_{t}^{--})(y_{1},y_{2}) \notag\\
&\quad+\la ^{-}(L_{2}^{--}k_{t}^{--})(y_{1},y_{2})+\la
^{-}a^{-}(y_{1}-y_{2})[k_{t}^{-}(y_{1})+k_{t}^{-}(y_{2})],
\label{eq2--}\\[3ex] \frac{\partial
k_{t}^{+-}(x,y)}{\partial t} &=(\la ^{+}+\la
^{-}-2)k_{t}^{+-}(x,y)+\la ^{+}L_{1}^{+-}k_{t}^{+-}(x,y)+\la
^{-}L_{2}^{+-}k_{t}^{+-}(x,y) \notag\\
&\quad+\la a(x-y)k_{t}^{-}(y)+\la \int_{\R^{d}}a(x-y'
)k_{t}^{--}(y,y' )dy', \label{eq2+-}\\[3ex]
\frac{\partial k_{t}^{++}(x_{1},x_{2})}{\partial t}
&=2k_{t}^{++}(x_{1},x_{2})(\la ^{+}-1)+\la
^{+}L_{1}^{++}k_{t}^{++}(x_{1},x_{2})+\la
^{+}L_{2}^{++}k_{t}^{++}(x_{1},x_{2}) \notag\\
&\quad+\{\la ^{+}a^{+}(x_{1}-x_{2})[k_{t}^{+}(x_{1})+k_{t}^{+}(x_{2})] \notag\\
&\quad+\la
\int_{\R^{d}}a(x_{1}-y)k_{t}^{+-}(x_{2},y)dy+\la
\int_{\R^{d}}a(x_{2}-y)k_{t}^{+-}(x_{1},y)dy\},\label{eq2++}
\end{align}
where Markov-type generators $L_{i}^{\pm\pm}$, $i=1,2$ are defined
on functions on $\X\times\X$ by
\begin{align*}
(L_{1}^{--}f)(y_{1},y_{2})
&=\int_{\R^{d}}a^{-}(y_{1}-y' )[f(y_{2},y' )-f(y_{2},y_{1})]dy' , \\
(L_{2}^{--}f)(y_{1},y_{2}) &=\int_{\R^{d}}a^{-}(y_{2}-y'
)[f(y_{1},y' )-f(y_{1},y_{2})]dy',\\
(L_{1}^{+-}f)(x,y) &=\int_{\R^{d}}a^{+}(x-x' )[f(x' ,y)-f(x,y)]dx' , \\
(L_{2}^{+-}f)(x,y) &=\int_{\R^{d}}a^{-}(y-y' )[f(x,y' )-f(x,y)]dy',\\
(L_{1}^{++}f)(x_{1},x_{2})
&=\int_{\R^{d}}a^{+}(x_{1}-x' )[f(x_{2},x' )-f(x_{2},x_{1})]dx',  \\
(L_{2}^{++}f)(x_{1},x_{2}) &=\int_{\R^{d}}a^{+}(x_{2}-x'
)[f(x_{1},x' )-f(x_{1},x_{2})]dx'.
\end{align*}%

Next propositions are direct corollaries of the perturbation method
(note also that any Markov semigroup preserves constants).
\begin{proposition}
The solutions of \eqref{eq1-}--\eqref{eq1+}
with initial values \eqref{initval1} have
the following forms:
\begin{align}
k_{t}^{-}(y) &=c^{-}e^{t(\la ^{-}-1)}+e^{t(\la
^{-}-1)}e^{t\la ^{-}L^{-}}\psi ^{-}(y), \label{expr_c1-}\\
k_{t}^{+}(x) &=c^{+}e^{t(\la ^{+}-1)}+e^{t(\la
^{+}-1)}e^{t\la ^{+}L^{+}}\psi ^{+}(x)+\la
c^{-}e^{t(\la^{+}-1)}\int_{0}^{t}e^{\tau(\la^{-}-
\la^{+})}d\tau \label{expr_c1+}\\
&\phantom{{}=c^{+}e^{t(\la ^{+}-1)}}+\la
e^{t(\la^{+}-1)}\int_{0}^{t}e^{\tau(\la^{-}-%
\la^{+})}e^{(t-\tau)\la^{+}L^{+}}(a\ast (e^{\tau
\la^{-}L^{-}}\psi^{-}))(x)d\tau.\notag
\end{align}
\end{proposition}

\begin{proposition} Let \eqref{tr_inv} holds.
Then the solutions
of~\eqref{eq2--}--\eqref{eq2++} with initial values \eqref{initval2} have
the following forms:
\begin{align}
k_{t}^{--}&(y_{1},y_{2}) =e^{t2(\la ^{-}-1)}e^{t\la
^{-}L_{1}^{--}}e^{t\la ^{-}L_{2}^{--}}(c^{--}+\varphi
^{--}(y_{1}-y_{2}))\notag
\\
&+\int_{0}^{t}e^{(t-\tau )2(\la ^{-}-1)}e^{(t-\tau )\la
^{-}L_{1}^{--}}e^{(t-\tau )\la ^{-}L_{2}^{--}}\la
^{-}a^{-}(y_{1}-y_{2})[k_{\tau }^{-}(y_{1})+k_{\tau
}^{-}(y_{2})]d\tau, \label{expr_c2--}\\[3ex]
k_{t}^{+-}&(x,y) =e^{t(\la ^{+}+\la ^{-}-2)}e^{t\la
^{+}L_{1}^{+-}}e^{t\la ^{-}L_{2}^{+-}}(c^{+-}+\varphi ^{+-}(x-y)) \notag\\
&+\int_{0}^{t}e^{(t-\tau )(\la ^{+}+\la
^{-}-2)}e^{(t-\tau
)\la ^{+}L_{1}^{+-}}e^{(t-\tau )\la ^{-}L_{2}^{+-}}\notag\\
&\qquad\times \{\la a(x-y)k_{\tau }^{-}(y)+\la \int_{\R^{d}}a(x-y'
)k_{\tau }^{--}(y,y' )dy' \}d\tau,
\label{expr_c2+-}\\[3ex] k_{t}^{++}&(x_{1},x_{2})
=e^{t2(\la ^{+}-1)}e^{t\la ^{+}L_{1}^{++}}e^{t\la
^{+}L_{2}^{++}}(c^{++}+\varphi
^{++}(x_{1}-x_{2}))\notag\\
&+\int_{0}^{t}e^{(t-\tau )2(\la ^{+}-1)}e^{(t-\tau
)\la ^{+}L_{1}^{++}}e^{(t-\tau )\la ^{+}L_{2}^{++}}\{\la
^{+}a^{+}(x_{1}-x_{2})[k_{\tau }^{+}(x_{1})+k_{\tau }^{+}(x_{2})] \notag\\
&+\la \int_{\R^{d}}a(x_{1}-y)k_{\tau
}^{+-}(x_{2},y)dy+\la
\int_{\R^{d}}a(x_{2}-y)k_{\tau
}^{+-}(x_{1},y)dy\}d\tau \label{expr_c2++}
\end{align}%
\end{proposition}

\subsection{Technical lemmas} In this subsection
we present several useful notations and notes
and prove
technical lemmas needed in the sequel.  Let us define\begin{alignat}{2}
\mu ^{+} &:=\la ^{+}-1,
&\qquad \mu ^{-} &:=\la ^{-}-1, \label{def_mu}\\
f^{+}(p)& :=\la ^{+}\hat{a}^{+}(p)-1,
&\qquad f^{-}(p)& :=\la ^{-}\hat{a}^{-}(p)-1.\label{def_f}
\end{alignat}

Note that conditions $0<\la^\pm\leq1$ equivalent
to $ -1<\mu ^{\pm }\leq 0$
and $\mu ^{\pm }=0$ only if $\la ^{\pm }=1$. Recall that $a^\pm$ are positive, even and
normalized. Then
\begin{equation}\label{apmineq}
\hat{a}^\pm(p)=\int_{\R^{d}}\cos (p,x)a^\pm(x)dx,
\qquad |\hat{a}^\pm(p)|\leq 1,
\end{equation}
and $\hat{a}^\pm(p)=1$
only at $p=0$.
Hence, the conditions $0<\la^\pm\leq1$ imply
\begin{equation}
-\la^{\pm}-1\leq f^{\pm }(p)\leq  \mu ^{\pm }\leq0,\label{flessmu}
\end{equation}
and $f^{\pm }(p)=\mu ^{\pm }$ only at point
$p=0$.

Let $C^-(\X)$ be a set of non-positive continuous  functions on $\X$ which equal to $0$ only
 on
 countable sets. Since Fourier image of
 integrable function is continuous one has
 $f^\pm\in C^-(\X)$. For any $f\in C^-(\X)$
 define two closed sets
\begin{equation}
\D[\pm]_f:=\{x\in\X :f(x)=f^\pm(x)\}.
\end{equation}
Note that that set $\X\setminus\D[+]_{f^-}=\X\setminus\D[-]_{f^+}$
has zero Lebesgue measure only if
$\la^+\hat{a}^+\equiv\la^-\hat{a}^{-}$ and, hence, $\la^+=\la^-$.

\begin{lemma}\label{int_a} Let $d\geq 3$ and $b\in L^1(\X)\cap
L^\infty(\X)$.Then
\begin{equation*}
c^\pm(p)=\frac{b(p)}{\hat{a}^\pm (p)-1}
\end{equation*}
are integrable functions on $\X$.
\end{lemma}
\begin{proof}By \eqref{apmineq}, $\hat{a}^\pm (0)=1$. Due to \eqref{decayofa}, $a^\pm$
has at least first and second finite moments.
Then using~\eqref{apmineq} one has
 in some neighbourhood of the origin
\begin{align*}
\hat{a}^\pm (p)-1 =\int_{\R^{d}}[\cos (p,x)-1]a^\pm(x)dx
\sim -\frac{1}{2}\int_{\R^{d}}( p,x)^{2}a^{\pm}(x)dx
\sim -\frac{1}{2} \vert p\vert^2
\end{align*}
and outside of this neighbourhood $\vert
\hat{a}^\pm (p)-1 \vert$ are bounded from
below.

Hence, $c^\pm$ are integrable in this neighbourhood
since $b$ is bounded and $\dfrac{1}{\vert p\vert^2}\in
L^1(\X)$ for $d\geq 3$; and $c^\pm$ are integrable outside of this neighbourhood since $b$ is
integrable.
\end{proof}

\begin{lemma}\label{boundint} Let $d\geq 3$, $0<\la^\pm\leq1$,
and $b\in L^1(\X)\cap
L^\infty(\X)$.Then
for any $f\in C^-(\X)$
\begin{equation*}
d^\pm(p)= b(p)\sup_{t\geq0} \frac{e^{tf^{}(p)}-e^{tf^\pm(p)}}
{f(p)-f^\pm(p)}
\end{equation*}
are integrable functions on $\X\setminus\D[\pm]_f$.
\end{lemma}
\begin{proof}
Let $p\in\X\setminus\D[+]_f$ for example.
Without loss of generality assume that $p\neq0$
and $f(p)\neq0$. Set $a=f(p), b=f^+(p)$.
Then $a<0$, $b<0$, $a\neq b$. Let us define
\begin{equation*}
h(t):=\dfrac{e^{ta}-e^{tb}}{a-b}, \quad
t\geq0.
\end{equation*}
Clearly, $h(t)\geq 0$ and $h(t)=0$ only at
$t=0$. One has
\begin{equation*}
h'(t):=\dfrac{be^{ta}\Bigl(\dfrac{a}{b}-e^{t(b-a)}\Bigr)}{a-b}.
\end{equation*}
Set $t_0=\dfrac{1}{b-a}\ln\dfrac{a}{b}$.
If $0>a>b$ then $t_0>0$ and for $0<t<t_0$
we have $e^{t(b-a)}>\dfrac{a}{b}$, hence,
$h'(t)>0$; for $t>t_0$ one has $h'(t)<0$.
If $0>b>a$ then $t_0>0$ also and for $0<t<t_0$
we obtain $e^{t(b-a)}<\dfrac{a}{b}$, therefore,
$h'(t)>0$; for $t>t_0$ again $h'(t)<0$. As~a~result,
\begin{equation*}
\max_{[0;\infty)}h(t)=h(t_0)=\dfrac{e^{t_0a}(1-e^{t_0(b-a)})}{a-b}=\frac{e^{t_0a}\Bigl(1-\dfrac{a}{b}\Bigr)}{a-b}=-\frac{1}{b}e^{t_0a}<-\frac{1}{b},
\end{equation*}
since $-b>0$, $a<0$.

Hence, for any $p\in\X\setminus\D[+]_f$,
$t\geq0$
\begin{equation*}
0\leq\frac{e^{tf^{}(p)}-e^{tf^+(p)}}
{f(p)-f^+(p)}<-\frac{1}{f^+(p)}.
\end{equation*}

Then using \eqref{flessmu}, \eqref{def_mu} for
$\la^+<1$ one has $\mu^+<0$ and
$d^+(p)< \dfrac{b(p)}{-\mu^+}$ that imply
the statement of this Lemma. For $\la^+=1$
the result is followed from Lemma~\ref{int_a}.
\end{proof}

\subsection{Asymptotic behaviour of the first
order correlation functions}
In this subsection we prove the Theorem~\ref{thm_as_1}.

\step{1} We should use~\eqref{expr_c1-}.
Note that $\psi^{-}\in L^1(\X)$ and Markov
semigroup maps $L^1(\X)$ into $L^1(\X)$.
Then using inverse Fourier transform one
has
\begin{equation}\label{invF1}
\bigl(e^{t\la ^{-}L^{-}}\psi^{-}\bigr)(y
) =
c_{d}\int_{\R^{d}}e^{i(p,y)}e^{t\la^{-}(\hat{a}^{-}(p)-1)}\hat{\psi}
^{-}\left( p\right) dp,
\end{equation}
where $c_d:=\dfrac{1}{(2\pi n)^d}$. Using
\eqref{apmineq}, the expression in the integral
in~\eqref{invF1} goes to $0$ for any $y$
and a.a. $p$. Since $\hat{\psi}^{-}\in L^1(\X)$
and $\left\vert e^{i(p,y)}e^{t\la^{-}(\hat{a}^{-}(p)-1)}\right\vert\leq
1$ one has that the~integral
also goes to $0$ for any $y$. Then the statement
is directly followed from~\eqref{expr_c1-}.

\step{2} We will use~\eqref{expr_c1+}. Note that similarly to the
first step $e^{t \la^{+}L^{+}}\psi^{+}\to 0$ point-wisely.

\step{2.1} If $\la ^{+}>1$ then for any
$\la^{-}>0$
\begin{equation*}
k_{t}^{+}\left( x\right) \rightarrow \infty ,
\end{equation*}%
since $\psi^{-}\geq \alpha^--c^{-}>-c^-$, hence, the last
term in~\eqref{expr_c1+} is bigger than
\begin{equation*}
-\la c^{-}
e^{t(\la^{+}-1)}\int_{0}^{t}e^{\tau(\la^{-}-\la^{+})}d\tau
\end{equation*}
and, therefore,
\begin{equation*}
k_{t}^{+}(x) >  c^{+}e^{t(\la ^{+}-1)}+e^{t(\la
^{+}-1)}e^{t\la ^{+}L^{+}}\psi ^{+}(x)\rightarrow\infty
\end{equation*}

\step{2.2} Let now $\la ^{+}\leq  1$. Divide proof on several
sub-steps. \step{2.2.1} Suppose $\la ^{+}=\la ^{-}=\nu$ then using~\eqref{expr_c1+} one
has
\begin{equation}\label{dop22}
k_{t}^{+}\left( x\right) =e^{t\left( \nu -1\right)
}c^{+}+e^{t\left( \nu
-1\right) }e^{t\nu L^{+}}\psi ^{+}\left( x\right) +\la e^{t\left( \nu -1\right) }c^{-}t + u_{t}(x)
\end{equation}
where
\begin{align*}
u_{t}(x)&=\la e^{t\left( \nu -1\right) }\int_{0}^{t}e^{\left( t-\tau
\right) \nu L^{+}}\left( a\ast (e^{\tau \nu L^{-}}\psi ^{-})\right)
\left( x\right) d\tau.
\end{align*}
Let us find $\lim\limits_{t\rightarrow\infty}u_{t}(x)$, for $\nu
\leq  1$. Note that $u_{t}\in L^1(\X)$ since semigroup and
convolution preserve integrability. Hence, we may compute the
Fourier transform of $u_t$:
\begin{equation}\label{hatut}
\hat{u}_{t}(p)= \left\{
\begin{array}{ll}
\la \hat{a}(p)\hat{\psi}^{-}(p)e^{tf^+(p)}t, & p\in\D[+]_{f^-},\\[2ex]
\la \hat{a}(p)\hat{\psi}^{-}(p)\dfrac{e^{tf^-(p)}-
e^{tf^+(p)}}
{f^{-}(p)-f^{+}(p)},  & p\in\X\setminus\D[+]_{f^-}.
\end{array}
\right.
\end{equation}

Since $\hat{\psi}^-$ is bounded and $\hat{a}$ is bounded and integrable due to \eqref{decayofa}
one can apply Lemma~\ref{boundint}, hence,
$\hat{u}_{t}(p)$ has integrable majorant
on $\X\setminus\D[+]_{f^-}$. Since $e^{ta}t<-\dfrac{e^{-1}}{a}$
for any $t\geq0$, $a<0$ one has for any $p\in\D[+]_{f^-}\setminus\{0\}$
\[
\bigl\vert\hat{u}_{t}(p)\bigr\vert\leq c_1\Biggl\vert
\frac{\hat{a}(p)}{f^+(p)}\Biggr\vert.
\]
Again if $\nu<1$ then denominator is separated
from zero, otherwise one can apply Lemma~\ref{int_a}.
As a result, $\hat{u}_{t}(p)$ has integrable majorant on whole $\X$ and pointwisely goes
to $0$ as $t\to\infty$ (except case $\nu=1$,
$p=0$).
Therefore, using majorized convergence theorem
the inverse Fourier transform of $\hat{u}_{t}(p)$ converges to zero, i.e. pointwisely $u_{t}(x)\to0$
as $t\to\infty$.

Thus, using~\eqref{dop22} one has that
$k_{t}^{+}\to\infty$ if $\nu=1$ and
$k_{t}^{+}\to0$ if $\nu<1$.

\step{2.2.2} Let now $\la ^{+}\neq \la
^{-}$. Using~\eqref{expr_c1+} obtain
\begin{align}
k_{t}^{+}\left( x\right) &= c^{+}e^{t\left( \la ^{+}-1\right)
}+e^{t\left( \la ^{+}-1\right) }e^{t\la ^{+}L^{+}}\psi
^{+}\left(
x\right) \notag\\
&\quad+\la c^{-}\frac{1}{\la ^{-}-\la ^{+}}\left( e^{t\left(
\la
^{-}-1\right) }-e^{t\left( \la ^{+}-1\right) }\right) \label{aaa}\\
&\quad+\la e^{t\left( \la ^{+}-1\right) }\int_{0}^{t}e^{\tau
\left( \la ^{-}-\la ^{+}\right) }e^{\left( t-\tau \right)
\la ^{+}L^{+}}\left( a\ast e^{\tau \la ^{-}L^{-}}\psi
^{-}\right) \left( x\right) d\tau .\notag
\end{align}

\step{2.2.2.1} Suppose that $\la ^{-}>1$.
Then since $\la ^{+}\leq 1$ and $\psi^-\geq\alpha^--c^-
>0$ we obtain that
\begin{equation*}
k_{t}^{+}\left( x\right) \rightarrow \infty ,\text{~~~}t\rightarrow \infty
\end{equation*}

\step{2.2.2.2} Next, let  $\la ^{-}<1,\ \la ^{+}<1$. Since $\hat{\psi}^-$ is bounded
one has for $M=\sup_\X\vert\hat{\psi}^-\vert$
that the last term in~\eqref{aaa} is not bigger (by absolute value) than
\begin{equation*}
 \frac{M}{\la^{-}-\la^{+}}\left(e^{t(%
\la^{-}-1)}-e^{t(\la^{+}-1)}\right)\to 0.
\end{equation*}
Then due to~\eqref{aaa} $k_{t}^{+}\left( x\right) \to 0$.

\step{2.2.2.3} Finally, let $\la ^{-}<1,~\la ^{+}=1$ or
$\la ^{-}=1,~\la ^{+}<1$. The last term in~\eqref{aaa}
is integrable function since semigroup and
convolution preserve integrability.
By direct computation its Fourier transform
has form~\eqref{hatut}. Hence, this last
term pointwisely goes to $0$.

As a result, by~\eqref{aaa} we obtain that
if $\la ^{+}=1$, $\la ^{-}<1$
\begin{equation*}
k_{t}^{+}\left( x\right) \rightarrow c^{+}+\frac{\la c^{-}}{1-\la
^{-}},\quad t\rightarrow \infty;
\end{equation*}
and if $\la ^{+}<1$, $\la ^{-}=1$%
\begin{equation*}
k_{t}^{+}\left( x\right) \rightarrow \frac{\la c^{-}}{1-\la ^{+}},\quad t\rightarrow \infty.
\end{equation*}

Theorem~\ref{thm_as_1} is proved.

\subsection{Asymptotic behaviour of the second
order correlation functions}
In this subsection we prove the Theorem~\ref{thm_as_2}.

First of all we present explicit expressions for $\Omega ^{++},~\Xi
^{--},~\Xi ^{+-},~\Xi ^{++}$, and after that we prove the Theorem. These functions are
inverse Fourier transforms of the following
\begin{align}
\omega ^{++}(p) &=\frac{\la ^{-}+\la -1}{\la ^{-}-1}\cdot \frac{%
c^{+}\hat{a}^{+}\left( p\right) }{1-\hat{a}^{+}\left( p\right) } ,\label{a1}\\
\xi ^{--}(p) &=\frac{c^{-}\hat{a}^{-}\left( p\right)
}{1-\hat{a}^{-}\left(
p\right) } ,\label{a2}\\
\xi ^{+-}(p) &=\frac{1}{2}\cdot \frac{\mu ^{-}+2}{2-\la ^{+}\hat{a}%
^{+}\left( p\right) -\hat{a}^{-}\left( p\right) }\cdot \frac{c^{-}\la
\hat{a}\left( p\right) }{1-\hat{a}^{-}\left( p\right) } ,\label{a3}\\
\xi ^{++}(p) &=\frac{\la }{1-\la ^{+}\hat{a}^{+}(p)}\left( \frac{%
\la ^{+}c^{-}\hat{a}^{+}(p)}{1-\la ^{+}}+\frac{\la c^{-}}{%
2-\la ^{+}\hat{a}^{+}(p)-\hat{a}^{+}(p)}\cdot \frac{\hat{a}^{2}(p)}{1-%
\hat{a}^{-}(p)}\right),\label{a4}
\end{align}
correspondingly.

Let us introduce the following denotations for the Markov semigroups
\begin{equation*}
T_{t}^{11}=e^{t\la ^{+}L_{1}^{++}},~~T_{t}^{12}=e^{t\la
^{+}L_{2}^{++}},~~T_{t}^{13}=e^{t\la ^{+}L_{1}^{+-}},
\end{equation*}%
\begin{equation*}
T_{t}^{21}=e^{t\la ^{-}L_{1}^{--}},~~T_{t}^{22}=e^{t\la
^{-}L_{2}^{--}},~~T_{t}^{23}=e^{t\la ^{-}L_{1}^{+-}}.
\end{equation*}

We start with trivial remark that for any
even functions $c,g\in L^1(\X)$
\begin{equation*}
\left( L_{1}g\right) \left( x_{1}-x_{2}\right) =\left( L_{2}g\right) \left(
x_{1}-x_{2}\right),
\end{equation*}
where
\begin{align*}
(L_{1}f)(x_{1},x_{2})
&:=\int_{\R^{d}}c(x_{1}-x' )[f(x_{2},x' )-f(x_{2},x_{1})]dx',  \\
(L_{2}f)(x_{1},x_{2}) &:=\int_{\R^{d}}c(x_{2}-x' )[f(x_{1},x'
)-f(x_{1},x_{2})]dx'.
\end{align*}

After transformations, substitutions and simplifying we obtain
for~\eqref{expr_c2--}--\eqref{expr_c2++}
the~following representations:
\begin{align*}
k_{t}^{--}(y_{1},y_{2})&=c^{--}e^{2\mu^{-}t}+e^{2%
\mu^{-}t}T_{t}^{21}T_{t}^{22}\varphi
^{--}(y_{1}-y_{2})+U_{t}^{--}(y_{1}-y_{2}),\\
k_{t}^{+-}(x,y) &=\left( c^{+-}-\frac{\la c^{--}}{\mu ^{-}-\mu ^{+}}%
\right) e^{(\mu^{+}+\mu^{-})t}+\frac{\la c^{--}}{\mu ^{-}-\mu ^{+}}%
e^{2\mu^{-}t} \\
&\quad+e^{(\mu^{+}+\mu^{-})t}T_{t}^{13}T_{t}^{23}\varphi
^{+-}(x-y)+U_{t}^{+-}(x-y),\\
k_{t}^{++}(x_{1},x_{2}) &=\left( c^{++}-\frac{2\la c^{+-}}{\mu
^{-}-\mu ^{+}}+\frac{\la ^{2}c^{--}}{(\mu ^{-}-\mu
^{+})^{2}}\right) e^{2\mu^{+}t}
\\
&\quad+\left( \frac{2\la c^{+-}}{\mu ^{-}-\mu ^{+}}-\frac{2\la ^{2}c^{--}%
}{(\mu ^{-}-\mu ^{+})^{2}}\right) e^{(\mu^{+}+\mu^{-})t} \\
&\quad+\frac{\la ^{2}c^{--}}{(\mu ^{-}-\mu^{+})^{2}}e^{2\mu^{-}t} \\
&\quad+e^{2\mu^{+}t}T_{t}^{11}T_{t}^{12}\varphi
^{++}(x_{1}-x_{2})+U_{t}^{++}(x_{1}-x_{2}).
\end{align*}
Here
\begin{equation*}
U_{t}^{--}(y_{1}-y_{2})=2\la ^{-}c^{-}\int_{0}^{t}e^{\mu ^{-}\tau
}e^{2\mu ^{-}(t-\tau )}T_{t-\tau }^{21}T_{t-\tau
}^{22}a^{-}(y_{1}-y_{2})d\tau,
\end{equation*}\\*[-2\parindent]
\begin{align*}
&U_{t}^{+-}(x-y) \\
&=\la c^{-}\int_{0}^{t}e^{\mu ^{-}\tau }e^{(\mu ^{+}+\mu
^{-})(t-\tau )}T_{t-\tau }^{13}T_{t-\tau }^{23}a(x-y)d\tau  \\
&\quad+\la \int_{0}^{t}e^{2\mu ^{-}\tau }e^{(\mu ^{+}+\mu
^{-})(t-\tau )}T_{t-\tau }^{13}T_{t-\tau }^{23}\int_{\R^{d}}a(x-y'
)T_{\tau }^{21}T_{\tau }^{22}\varphi
^{--}(y-y' )dy' d\tau  \\
&\quad+2c^{-}\la \la ^{-}\int_{0}^{t}e^{(\mu ^{+}+\mu ^{-})(t-\tau
)}T_{t-\tau }^{13}T_{t-\tau }^{23}\\&\qquad\times\int_{\R^{d}}a(x-y'
)\int_{0}^{\tau }e^{\mu ^{-}s}e^{2\mu ^{-}(\tau -s)}T_{\tau
-s}^{21}T_{\tau -s}^{22}a^{-}(y-y' )dsdy' d\tau,
\end{align*}\\*[-2\parindent]
\begin{align*}
&U_{t}^{++}(x_{1}-x_{2}) \\
&=2\la ^{+}c^{+}\int_{0}^{t}e^{\mu ^{+}\tau }e^{2\mu
^{+}(t-\tau
)}T_{t-\tau }^{11}T_{t-\tau }^{12}a^{+}(x_{1}-x_{2})d\tau  \\
&\quad+2\la \la ^{+}c^{-}\int_{0}^{t}e^{2\mu ^{+}(t-\tau
)}T_{t-\tau }^{11}T_{t-\tau
}^{12}a^{+}(x_{1}-x_{2})\int_{0}^{\tau
}e^{\mu ^{-}s}e^{\mu ^{+}(\tau -s)}dsd\tau  \\
&\quad+2\la \int_{0}^{t}e^{(\mu ^{+}+\mu ^{-})\tau }e^{2\mu
^{+}(t-\tau )}T_{t-\tau }^{11}T_{t-\tau
}^{12}\int_{\R^{d}}a(x_{1}-y)T_{\tau }^{13}T_{\tau
}^{23}\varphi
^{+-}(x_{2}-y)dyd\tau  \\
&\quad+2\la ^{2}c^{-}\int_{0}^{t}e^{2\mu ^{+}(t-\tau
)}T_{t-\tau }^{11}T_{t-\tau
}^{12}\int_{\R^{d}}a(x_{1}-y)\\&\qquad\times\int_{0}^{\tau
}e^{\mu ^{-}s}e^{(\mu ^{+}+\mu ^{-})(\tau -s)}T_{\tau
-s}^{13}T_{\tau
-s}^{23}a(x_{2}-y)dydsd\tau  \\
&\quad+2\la ^{2}\int_{0}^{t}e^{2\mu ^{+}(t-\tau )}T_{t-\tau
}^{11}T_{t-\tau
}^{12}\int_{\R^{d}}a(x_{1}-y)\int_{0}^{\tau
}e^{2\mu ^{-}s}e^{(\mu ^{+}+\mu ^{-})(\tau -s)}T_{\tau
-s}^{13}T_{\tau
-s}^{23} \\
&\qquad \times \int_{\R^{d}}a(x_{2}-y' )T_{s}^{21}T_{s}^{22}\varphi ^{--}(y-y' )dy' dsdyd\tau  \\
&\quad+4\la ^{-}c^{-}\la ^{2}\int_{0}^{t}e^{2\mu
^{+}(t-\tau )}T_{t-\tau }^{11}T_{t-\tau
}^{12}\int_{\R^{d}}a(x_{1}-y)\\&\qquad\times\int_{0}^{\tau
}e^{(\mu ^{+}+\mu ^{-})(\tau -s)}T_{\tau -s}^{13}T_{\tau
-s}^{23}\int_{\R^{d}}a(x_{2}-y' ) \\
&\qquad \times \int_{0}^{s}e^{\mu ^{-}\theta }e^{2\mu ^{-}(s-\theta
)}T_{s-\theta }^{21}T_{s-\theta }^{22}a^{-}(y-y' )d\theta dy'
dsdyd\tau.
\end{align*}

Since semigroups and convolutions preserve integrability we have that $T_{t}^{21}T_{t}^{22}\varphi^{--}$,
$T_{t}^{13}T_{t}^{23}\varphi^{+-}$, $T_{t}^{11}T_{t}^{12}\varphi^{++}$
as well as $U_{t}^{--},~U_{t}^{+-}$ and $U_{t}^{++}$ are integrable on $\R^{d}$ functions. So, to find their limits as $t\to\infty$ we may use the Fourier transforms.

Namely,\begin{align*}
T_{t}^{21}T_{t}^{22}\varphi
^{--}(y_1-y_2)&=c_{d}\int%
\limits_{\R^{d}}e^{ip(y_1-y_2)}e^{2(f^{-}(p)-\mu^-)t}\hat{\varphi
}^{--}(p)dp,
\\
T_{t}^{13}T_{t}^{23}\varphi ^{+-}(x-y)&=c_{d}\int%
\limits_{\R^{d}}e^{ip(x-y)}e^{(f^{+}(p)-\mu^+)t}
e^{(f^{-}(p)-\mu^-)t}\hat{\varphi} ^{+-}(p)dp,
\\
T_{t}^{11}T_{t}^{12}\varphi ^{++}(x_1-x_2)&=c_{d}\int%
\limits_{\R^{d}}e^{ip(x_1-x_2)}e^{2(f^{+}(p)-\mu^+)t}
\hat{\varphi }^{++}(p)dp.
\end{align*}
Since $\hat{\varphi}^{--}, \hat{\varphi}^{+-},
\hat{\varphi}^{++}$ are integrable we have
using \eqref{flessmu} and dominated convergence theorem that these three terms go to $0$.

Let us introduce for further simplicity of notations the following functions
\begin{align*}
h_1(p)&:=\mu^+-2f^+(p)\geq 0, \\
h_2(p)&:=\mu^--2f^-(p)\geq 0, \\
h_3(p)&:=f^+(p)+f^-(p)<0, \\
h_4(p)&:=\mu^--f^+(p)-f^-(p)\geq 0.
\end{align*}
These inequalities are followed from~\eqref{def_mu},
\eqref{def_f} and \eqref{flessmu} as well
as the fact that equalities are possible
only at $p=0$.

Consider also the following two functions $g_1$ and $g_2$
\begin{align*}
g_1(p)&=f^-(p)-f^+(p), \\
g_2(p)&=\mu^--2f^+(p).
\end{align*}
They can be equal zero on a~set of non-zero
measure.

We have in the new notations:
\begin{align*}
\widehat{U_{t}}^{{--}}(p)&=2c^-\la^-\hat{a}^-(p)e^{2f^-(p)t}\int%
\limits_0^te^{h_2(p)\tau}d\tau,
\\
\widehat{U_{t}}^{{+-}}(p)&=c^-\la\hat{a}(p)e^{h_3(p)t}\int%
\limits_0^te^{h_4(p)\tau}d\tau \\
&\quad+\la\hat{a}(p)\hat{\varphi}^{--}(p)e^{h_3(p)t}\int_0^te^{g_1(p)%
\tau}d\tau \\
&\quad+2c^-\la\hat{a}(p)\la^-\hat{a}^-(p)e^{h_3(p)t}\int%
\limits_0^te^{g_1(p)\tau}\int_0^\tau e^{h_2(p)s}dsd\tau,
\\
\widehat{U_{t}}^{{++}}(p)&=2c^+\la^+\hat{a}^+(p)e^{2f^+(p)t}\int%
\limits_0^te^{h_1(p)\tau}d\tau \\
&\quad+\frac{2\la c^-\la^+\hat{a}^+(p)}{\mu^--\mu^+}e^{2f^+(p)t}\left(%
\int_0^te^{g_2(p)\tau}d\tau-\int_0^te^{h_1(p)\tau}d\tau\right)
\\
&\quad+2\la\hat{a}(p)\hat{\varphi}^{+-}(p)e^{2f^+(p)t}\int%
\limits_0^te^{g_1(p)\tau}d\tau \\
&\quad+2c^-\la^2\hat{a}^2(p)e^{2f^+(p)t}\int_0^te^{g_1(p)\tau}\int%
\limits_0^{\tau}e^{h_4(p)s}dsd\tau \\
&\quad+2\la^2\hat{a}^2(p)\hat{\varphi}^{+-}(p)e^{2f^+(p)t}\int%
\limits_0^te^{g_1(p)\tau}\int_0^{\tau}e^{g_1(p)s}dsd\tau \\
&\quad+4c^-\la^-\hat{a}^-(p)\la^2\hat{a}^2(p)e^{2f^+(p)t}\int%
\limits_0^te^{g_1(p)\tau}\int_0^{\tau}e^{g_1(p)s}\int%
\limits_0^se^{h_2(p)\theta}d\theta dsd\tau.
\end{align*}

Let us consider the following closed set $\D=\D_1\cup\D_2$,
where $\D_1:=\{p:\ g_1(p)=0\}=\D[+]_{f^-},\ \mathfrak{D}_2=\{p:\ g_2(p)=0\}$.
It's easy to see that $\D_1\cap\D_2=\varnothing$.
Indeed, by~\eqref{flessmu} for any $p\in\D_1\cap\D_2$
\[
\mu^-=2f^+(p)=2f^-(p)\leq 2\mu^-.
\]
But $\mu^-\leq0$, hence, it should be equality
that implies
$f^-(p)=\mu^-$, and with necessity $p=0$.
But if $0\in\D_1\cap\D_2$, then $f^+(0)=f^-(0)$,
i.e., $\mu^+=\mu^-$, that contradicts to the condition of the theorem.

Next we note that the functions  $\widehat{U_{t}}^{{+-}}(p)$ and $%
\widehat{U_{t}}^{{++}}(p)$ have different
explicit expressions for $p\in\D$ and for
$p\in\DC:=\R^d\setminus\D$. Note also that these functions are continuous functions of $p$ as compositions of the integrals of
the continuous functions of $t$ with continuous dependence on a parameter $p$.
Hence, for calculate these expressions for $p\in\D$ we may calculate their for $p\in\DC$
and take limits
as $\dist(p,\D)\to 0$.

By direct calculations for any $p\in\DC\setminus\{0\}$
we obtain
\begin{align*}
\widehat{U}_{t}^{--}(p)&=2\lambda
^{-}c^{-}\hat{a}^{-}(p)\frac{e^{\mu ^{-}t}-e^{2f^{-}(p)t}}{\mu
^{-}-2f^{-}(p)},
\\
\widehat{U}_{t}^{+-}(p)&= \lambda c^{-}\hat{a}(p)\cdot \frac{\mu ^{-}+2}{\mu
^{-}-2f^{-}(p)}\cdot \frac{e^{\mu ^{-}t}-e^{[f^{+}(p)+f^{-}(p)]t}}{\mu
^{-}-[f^{+}(p)+f^{-}(p)]} \\
& \quad+\left( \lambda \hat{a}(p)\hat{\varphi}^{--}(p)-\frac{2c^{-}\lambda
\lambda ^{-}\hat{a}(p)\hat{a}^{-}(p)}{\mu ^{-}-2f^{-}(p)}\right)
G_{t}^{(1)}(p) e^{2f^{-}(p)t},
\\
\widehat{U}_{t}^{++}(p)&=\left(\frac{2\lambda c^{-}\lambda^{+}\hat{a}^{+}(p)%
}{\mu^{-}-\mu^{+}}+\frac{2c^{-}\lambda^{2}\hat{a}^{2}(p)}{%
\mu^{-}-f^{+}(p)-f^{-}(p)}\cdot\frac{\mu^{-}+2}{\mu^{-}-2f^{-}(p)}%
\right)G_t^{(2)}(p) e^{2f^{+}(p)t} \\
&\quad+2c^{+}\lambda^{+}\hat{a}^{+}(p)\cdot\frac{\mu^{-}-\mu^{+}+\lambda}{%
\mu^{-}-\mu^{+}}\cdot\frac{e^{\mu^{+}t}-e^{2f^{+}(p)t}}{\mu^{+}-2f^{+}(p)} \\
&\quad+\left(\lambda^{2}\hat{a}^{2}(p)\hat{\varphi}^{--}(p)-\frac{%
2c^{-}\lambda^{-}\hat{a}^{-}(p)\lambda^{2}\hat{a}^{2}(p)}{\mu^{-}-2f^{-}(p)}%
\right)\left(G_t^{(1^{})}(p)\right)^{2} e^{2f^{-}(p)t} \\
&\quad+\left(2\lambda\hat{a}(p)\hat{\varphi}^{+-}(p)-\frac{2c^{-}\lambda^{2}\hat{a%
}^{2}(p)}{\mu^{-}-f^{+}(p)-f^{-}(p)}\cdot\frac{\mu^{-}+2}{\mu^{-}-2f^{-}(p)}%
\right)\\ &\qquad \times G_t^{(1^{})}(p) e^{[f^{+}(p)+f^{-}(p)]t},
\end{align*}
where we denote objects which are not defined
for $p\in\D$ by
\begin{align*}
&&G_t^{(1)}(p)&=\frac{e^{[f^+(p)-f^-(p)]t}-1}{f^+(p)-f^-(p)},
&p&\in\DC_1:=\X\setminus\D_1,\\
&&G_t^{(2)}(p)&=\frac{e^{[\mu^--2f^+(p)]t}-1}{\mu^--2f^+(p)},
&p&\in\DC_2:=\X\setminus\D_2.
\end{align*}

Obviously $\dist(p,\D_1)\to 0$ implies $g_1(p)\to0$
and, hence, $G_{t}^{(1)}(p)\to t$. In the
same manner $\dist(p,\D_2)\to 0$ provides  $G_{t}^{(2)}(p)\to t$. Therefore, for obtain
the explicit expressions for $\widehat{U_{t}}^{{+-}}(p)$ and $\widehat{U_{t}}^{{++}}(p)$ on $\D\setminus\{0\}$
it's enough to define
\[
G_t^{(1)}(p):=t, \quad
p\in\D_1; \qquad G_t^{(2)}(p):=t, \quad
p\in\D_2.
\]

Then we have for any $b\in L^1(\X)\cap L^\infty(\X)$
\begin{equation*}
\vert b(p) \vert G_{t}^{(1)}(p) e^{f^{-}(p)t}\leq
\begin{cases}
\vert b(p)\vert \dfrac{e^{f^+(p)t}-e^{f^-(p)t}}%
{f^+(p)-f^-(p)},\qquad p\in\DC_1\setminus\{0\},\\
\vert b(p)\vert\dfrac{ e^{-1}}{-2f^-(p)}, \qquad\qquad\quad\, p\in\D_1.
\end{cases}
\end{equation*}
And by result and proof of Lemma~\ref{boundint}
this function has integrable majorante (which
doesn't depend on $t$) on
whole $\X$.
Note also that $e^{f^\pm(p)t}\leq 1$, hence,
all terms with $G_t^{(1)}$ have this property.

Next,
\begin{equation*}
\vert b(p) \vert G_{t}^{(2)}(p) e^{2f^{+}(p)t}\leq
\begin{cases}
\vert b(p)\vert \dfrac{e^{\mu^-t}- e^{2f^{+}(p)t}}{\mu^--2f^+(p)}, \quad p\in\DC_2\setminus\{0\},\\
\vert b(p)\vert \dfrac{e^{-1}}{-2f^+(p)},\qquad\quad\
\,\,
p\in\D_2.
\end{cases}
\end{equation*}
If $\mu^-<0$ then may apply the previous
considerations ($\mu^-\in C^-$). Otherwise, we may use that
a function $u(t)=\dfrac{1-e^{at}}{-a}$ ($a<0$)
is increasing and, hence, bounded by $u(+\infty)=-
\dfrac{1}{a}.$

Note also that other numerators depended on $t$ in the expressions for $\widehat{U}_t^{--}$,
$\widehat{U}_t^{+-}$, $\widehat{U}_t^{++}$
may be estimated by $2$ (recall that corresponding
denominators are not equal to $0$ if $p\neq0$).

Therefore, for prove that functions $\widehat{U}_t^{--}$,
$\widehat{U}_t^{+-}$, $\widehat{U}_t^{++}$
have integrable majorants it's enough to
show that all terms which independent on $t$ are
integrable.
 Recall that  $\hat{\varphi}^{--}$, $\hat{%
\varphi}^{+-}$ and $\hat{\varphi}^{++}$ are bounded, $\hat{a%
},\ \hat{a}^{+}$ and $\hat{a}^{-}$ are bounded and integrable. Thus, we should prove integrability
of two terms:
\begin{equation}\label{lala}
\frac{b(p)}{\mu^{\pm}-2f^{\pm}(p)}
\quad \text{and}\quad  \frac{%
b(p)}{\mu^{-}-f^{-}(p)-f^{+}(p)}\cdot
\frac{1}{\mu^{-}-2f^{-}(p)},
\end{equation}
where $b\in L^1(\X)\cap L^\infty(\X)$.

If $\mu^{\pm}=0$ then we have
\begin{align*}
\frac{b(p)}{\mu^{\pm}-2f^{\pm}(p)}=-\frac{1}{2}\frac{b(p)}{\hat{a}^{\pm}(p)-1}
\end{align*}
and due to Lemma~\ref{int_a} these functions
are integrable.  If $\mu^{\pm}< 0$ then using
\eqref{flessmu} we obtain
\begin{align*}
0<-\mu^{\pm} \leq  \mu^{\pm}-2f^{\pm}(p),
\end{align*}
that implies
\[
\frac{\vert b(p)\vert}{\mu^{\pm}-2f^{\pm}(p)} \leq\frac{\vert
b(p)\vert}{-\mu^\pm}
\]
which are also integrable functions.

Next, if $\mu^{-}=0$ then $\mu^{+}<0$ and
using \eqref{flessmu}
\[
(\mu^{-}-f^{-}(p)-f^{+}(p))(\mu^{-}-2f^{-}(p))\geq
-2\mu^+(1-\hat{a}^{-}(p)),
\]
and we again may use Lemma~\ref{int_a}.
Finally, if $\mu^{-}<0$ then $\mu^+=0$ and
\[
\bigl( (\mu^{-}-f^{-}(p))+(-f^{+}(p))\bigr) \cdot
\bigl(\mu^{-}-2f^{-}(p)\bigr) \geq -\mu^-
(1-\hat{a}^{+}(p)),
\]
and we also may use Lemma~\ref{int_a}.

As a result, the functions $\widehat{U}_t^{--}$,
$\widehat{U}_t^{+-}$, $\widehat{U}_t^{++}$ have integrable majorants
and by dominated convergence theorem for obtain limits of
${U}_t^{--}$, ${U}_t^{+-}$, ${U}_t^{++}$ as $t\to\infty$ we may
calculate limits of the Fourier transforms and after apply the
inverse Fourier transforms. Hence, taking $t\to\infty$ in the
expressions for $\widehat{U}_t^{--}$, $\widehat{U}_t^{+-}$,
$\widehat{U}_t^{++}$ we immediately obtain the statement of the
Theorem~\ref{thm_as_2} with functions $\Omega ^{++},~\Xi ^{--},~\Xi
^{+-},~\Xi ^{++}$ which are inverse Fourier transforms of
\eqref{a1}--\eqref{a4}.

\subsection*{Acknowledgments}
The authors acknowledge the financial support of the DFG through SFB
701 ``Spectral structures and topological methods in mathematics'',
Bielefeld University.

\end{document}